\RequirePackage{fix-cm}
\RequirePackage{amsmath}

\documentclass{elsarticle}

\usepackage{latexsym,graphicx,xcolor}
\usepackage{amsmath}

\usepackage{amssymb}
\usepackage{float}
\usepackage{amsthm}
\usepackage[english]{babel}
\usepackage{enumitem}
\usepackage[utf8]{inputenc}
\usepackage[T1]{fontenc}
\usepackage{tikz}
\usetikzlibrary{trees}
\usetikzlibrary{positioning}
\usepackage[normalem]{ulem}
\usepackage{soul}
\usepackage{natbib}
\usepackage{multirow}
\usepackage{makecell, siunitx}

\newtheorem{theorem}{Theorem}
\newtheorem{lemma}{Lemma}
\newtheorem{definition}{Definition}
\newtheorem{prop}{Proposition}
\newtheorem{corollary}{Corollary}

\newtheorem{remark}{Remark}

\newtheorem{conj}{Conjecture}
\newtheorem{claims}{Claim}

\newcommand{\e}{\epsilon}
\newcommand{\R}{\mathbb{R}}
\newcommand{\N}{\mathbb{N}}
\renewcommand{\t}{\tau}
\renewcommand{\k}{\kappa}
\renewcommand{\o}{\omega}
\renewcommand{\d}{\delta}

\renewcommand{\Re}{\operatorname{Re}}
\renewcommand{\Im}{\operatorname{Im}}

\def\TWOLF{\texttt{OPTiWinD}}
\def\OPTiWinD{\texttt{OPTiWinD}}

\begin{document}
\begin{frontmatter}
\title{Orienteering problem with time-windows\\ and updating delay }

\author[1]{Marc Demange}
\ead{marc.demange@rmit.edu.au} 
\author[1]{David Ellison} 
\ead{david.ellison2@rmit.edu.au} 
\author[3]{Bertrand Jouve}
\ead{jouve@univ-tlse2.fr} 



\address[1]{RMIT University, School of Science, Melbourne, Australia}

 \address[3]  {
              LISST UMR5193, Toulouse University, CNRS, France}


\begin{abstract}
    The Orienteering Problem with Time Window and Delay (\OPTiWinD) is a variant of the online orienteering problem. A series of requests appear in various locations while a vehicle moves within the territory to serve them. Each request has a time window during which it can be served and a weight which describes its importance. There is also a minimum delay $T$ between successive requests. The objective is to find a path for the vehicles that maximises the sum of the weights of the requests served. We further assume that the length of each time window is equal to the diameter of the territory. We study the optimal performance and competitive ratio for the set of instances with $n$ requests. We obtain complete resolution for $T$ at least half of the diameter, small values of $T$ or small values of $n$, as well as partial results in the remaining cases.
\end{abstract}
\end{frontmatter}

\providecommand{\keywords}[1]{\textit{\textbf{Keywords --- }#1}}

\section{Introduction and related work}
\subsection{Defining the problem}
This work is motivated by a wildfire emergency management problem. On very hot days and during the dry season, in regions exposed to wildfires, it is not rare to observe, in a single region, up to one hundred  fire ignitions per day. Some of these fires spread very quickly and, if not extinguished quickly, they escape and become completely out of control. Although the risk factors, such as temperature, wind, fuel load, dryness, etc. and the main causes of ignition are well-understood, it is almost impossible to forecast, during a day, where and when a fire will occur.  Wildfires are already an important threat to human lives, goods and the environment; and in the near future, this is expected to become worse due to  climate change. 

For emergency management commanders, one of the main challenges is  to rely on efficient decision making. And when disaster strikes, especially when fatalities occur, they need to be able to justify their decisions once the crisis is over. Hence, in order to satisfy legal constraints, the decision chain needs to be very clear and well-defined.
The most critical decision commanders need to make during the initial outbreak is the allocation of firefighting resources in cases where there are not enough resources to keep all the fires under control. Such a decision relies on the information available about the on-going fires and the threats they pose. 
Early warning systems are set-up to detect new fires, but once a smoke plume has been detected, an on-site reconnaissance is necessary to assess the importance of the fire, evaluate the possible induced risk and  estimate how long until the fire can no longer be contained. Based on this information, the management team decides whether or not to send firefighting resources to each new fire. Given that making this decision requires having evaluated all the risks associated with the fire, each new fire is only taken into account after a delay corresponding to the time used by the reconnaissance team. We will assume that there is a single firefighter team which cannot be divided, thus two fires igniting simultaneously will be processed one at a time with a delay between them. Natural generalisations may induce several or splittable teams. The objective is to schedule the movements of the team by selecting which fires to extinguish and when, in order to maximise the total weight of the fires contained. 
We model this situation using a sequence of fires, each with a weight representing the expected loss if the fire goes out of control and a time window in which the fire has been assessed and remains controllable. We include a minimum delay between the release time of fires.

This model corresponds to the standard Orienteering Problem in a  space $X$, with Time Windows and with the addition of time delays. This variant will be called {\em The Orienteering Problem with Time Window and Delay} (\OPTiWinD). In this problem, time is continuous and new requests can appear at any moment and at any point in $X$. Each request has a location, a weight and a time window during which it can be served. There is a minimum delay between successive new requests. The player controls a unique vehicle which travels to serve the requests. A request is instantaneously satisfied when it is visited. The player's objective is to find a path in the space that maximises the sum of the weights of the requests served. A strategy of the player is described as an algorithm.  

We will consider online algorithms for instances where each request is revealed at the start of its time window. This corresponds to the reality of our motivating problem, for which the difficulty to establish firefighting strategies is largely due to not knowing in advance when and where the fires will appear.

\subsection{Related work}

The underlying combinatorial optimisation problem  is a natural generalisation of the Metric Travelling Salesman Problem (Metric TSP), which is also called the \emph{Orienteering Problem}~\citep{orienteering1}: one is given a time limit  and a graph with $n$ vertices including departure and arrival points. Each vertex is associated with a score and each edge with a distance. The objective is to select some vertices and an order to visit them while travelling at constant speed~1 from the departure point to the arrival point, so that the time limit is not exceeded and the total score of visited vertices is maximised. In \citet{laporte1990selective}, the orienteering problem is called \emph{Selective TSP}. The same problem has been studied under other names, in particular  \emph{Bank Robber Problem}~\cite{quota} and \emph{Maximum Collection Problem}. The dual problem where the aim is to collect a fixed targeted total score in the minimum time is  called \emph{Quota TSP}~\citep{quota} or also \emph{Prize Collecting Problem}. These problems fall into the class of \emph{Travelling Salesman Problems with Profits}~\citep{tsp-profit}. 
Our application deals with the version of the orienteering problem with time windows~\citep{orienteeringsurv}: instead of an overall time limit, each vertex is associated with a release time and a deadline and it must be visited between these two dates in order to collect the associated score.  

Some of these problems have been considered in the online case, also referred as \emph{dynamic case}. In most cases, the online setup consists in revealing requests (vertices) over time, so that the online algorithm needs to decide about the movements of the vehicle without knowing when and where the next request will be revealed. The Metric Travelling Salesman has been considered under this setup  in~\cite{olqtsp,Ausiello2001}, for instance, and the Quota TSP in~\cite{olqtsp}.
When the delay is zero, $\OPTiWinD$ is equivalent to the \emph{Whack-a-Mole Problem}, which is defined in~\citet{gut06}, or the \emph{Dynamic Traveling Repair Problem}~\citep{iri04}. For a space which is a truncated line $[-L,L]$, \citet{gut06} note that the only non trivial cases are those for $\theta/4 < L \leq \theta$, where $\theta$ denotes the size of the time windows. Thus, they study the case $\theta=L$ with moles that may only arrive at integer positions on $[-L,L]$. In the case of the segment, we study in this paper the case where $L=\theta/2$,  which is not examined in  \citet{gut06}. To our knowledge, no previous paper deals with a delay between requests.

\section{Preliminaries}

\subsection{Geodesic metric spaces}

A metric space $E=(X,d)$ is said to be \emph{geodesic}~\citep{dez09} if for any $a,b\in X$, there is a continuous path $\gamma:[0,1]\rightarrow X$ from $a$ to $b$ such that the range of $\gamma$ is isometric to the segment $[0,d(a,b)]$.

Defining the offline orienteering problem on graphs or on metric spaces is equivalent. The distance, or travel time, between requests is the only aspect of the underlying metric space relevant to the game. However, in the online case, the position of the vehicle must be defined at all time. Hence the online orienteering problem is defined either on graphs (e.g. \citep{iri04}) or on geodesic metric spaces (e.g. \citep{Ausiello2001}), producing two slightly different online games.

Having a geodesic metric space means that the distance between two points is equal to the length of the shortest path connecting them. The vehicle is allowed to change directions whenever a new request appears, even when the vehicle is travelling between requests. Alternatively, on graphs, the vehicle can only select a new direction after reaching a vertex.

\subsection{Definitions and notations}
 For geodesic metric spaces with only one pair of points maximising the distance, it is natural to consider a worst case scenario when all the requests appear on these two points. As it happens, the case where $E$ is a segment is critical to understanding the general case. Thus, except in Section \ref{sec:geodesic}, we will consider the online problem on the segment $[-1,1]$ with the vehicle starting at $0$. We will show in Section \ref{sec:geodesic} that most of the results for the segment can be extended to bounded geodesic metric spaces of diameter 2, and we will give counter-examples to the others.

An instance of \OPTiWinD is a set of requests, each characterised by a triple $(A,[r,r'],\delta)$ where $A\in X$ is the location of the request, $r$ is the release time, $r'$ is the time when the request is lost if it is not served and $\delta \in \R_+^*$ is the \emph{weight} (or score) associated with the request.

A request $(A,[r,r'],\delta)$ is said to be \emph{served} if the vehicle is on $A$ at a time $t\in [r, r']$. A feasible solution is an itinerary for the vehicle starting at $O$ at time $t=0$ and serving some requests. The objective is to maximise the total weight of the requests served. We denote by $\mathcal I_{n,T}$ the set of instances of $\OPTiWinD$ on $[-1,1]$ with at most $n$ requests and delay at least $T$ between requests.

We will make the following assumptions about  the $\OPTiWinD$ Problem:  
\begin{itemize}
    \item We consider that $X$ is the segment $[-1,1]$, except in Section~\ref{sec:geodesic} where we extend the problem to the case of a bounded geodesic metric space $X$ of diameter 2.

    \item The vehicle is able to move no faster than unit speed in $X$ and a request is instantaneously satisfied when it is visited. 

    \item  We assume that each time window is of the form $[r,r+2]$ where $r$ is the release time. Thus, once a request appears, the vehicle always has a chance to reach its location in time from anywhere in $E$.
    
    \item The release times of any two successive requests differ by at least a delay $T\geq 0$, which is independent from the requests.
     
    \item Online algorithms for $\OPTiWinD$  learn about a new request (its location and weight) at the start of its associated time window. The vehicle can serve it at any time within that time window or it can choose not to serve it at all. The geodesic metric space is completely known from the beginning. 
    
\end{itemize}

In the following, the sets of non-negative integers and real numbers are denoted by $\N$ and $\R$ respectively. The golden ratio is $\varphi$. The $i$-th request is denoted $f_i$ and its release time is $t_i$. The weight of $f_i$ is $\d_i$ and $S_i=\sum_{k=1}^i \d_k$. The earliest time when the vehicle may reach $f_i$, assuming it starts heading towards it at $t_i$, is $\t_i$.


\subsection{Performance and competitive ratio}

The \emph{performance} $\lambda_{ALG}^I$ of an algorithm $ALG$ for an instance $I\in\mathcal I_{n,T}$ of  $\TWOLF$  is the ratio of the sum of the weights of the served requests by the total sum of the weights of all requests. The performance of $ALG$ for at most $n$ requests and delay $T$ is then defined as:
$$\lambda_{ALG}^{\mathcal I_{n,T}}=\inf_{I\in\mathcal I_{n,T}} \lambda_{ALG}^I. $$
By definition, any $\lambda_{ALG}^I$ is in $[0,1]$; hence $\lambda_{ALG}^{\mathcal I_{n,T}} \in [0,1]$. An online algorithm $ALG$ is $\gamma$-competitive for problem $\TWOLF$ with at most $n$ requests and delay $T$ if, given any instance $I\in \mathcal I_{n,T}$, the performance of the online algorithm is at least $\gamma$ multiplied by the performance of $OPT$, the optimal offline algorithm:  $\forall I\in \mathcal I_{n,T}$,  $\lambda_{ALG}^I \geq \gamma \cdot \lambda_{OPT}^I$, with $\gamma>0$. Note that this definition of competitive ratio is standard for maximisation problems. Hence, we have $0\leq \gamma \leq 1$. The {\emph competitive ratio} of $ALG$ for at most $n$ requests and delay $T$ is defined by: $$\gamma_{ALG}^{\mathcal I_{n,T}}=\inf_{I\in\mathcal I_{n,T}} \frac{\lambda_{ALG}^I}{\lambda_{OPT}^I}.$$



For any algorithm $ALG$, we have $\lambda_{ALG}^{\mathcal I_{n,T}} \leq \gamma_{ALG}^{\mathcal I_{n,T}}$. In order to evaluate the performance and competitive ratio of an algorithm, we need to consider a worst case scenario. This situation corresponds to having a malicious adversary deciding  when the requests are released and choosing their weights. While our problem is a one-player online game, as the adversary is not technically a player, it is standard practice to discuss strategies of the adversary as though in a two-player game.

The \emph{optimal performance} and the \emph{optimal competitive ratio} are respectively:
$$\lambda^{\mathcal I_{n,T}}=\sup_{ALG} \lambda_{ALG}^{\mathcal I_{n,T}} ~~~\mathrm{and} ~~~\gamma^{\mathcal I_{n,T}}=\sup_{ALG} \gamma_{ALG}^{\mathcal I_{n,T}} $$

An online algorithm will be called \emph{optimal} if its competitive ratio is equal to the optimal competitive ratio. Note that we have:
$$\lambda^{\mathcal I_{n,T}} \leq \gamma^{\mathcal I_{n,T}}.$$

\begin{definition}\label{f(n)-comp}
Given a function $f:\N\rightarrow \R^+$, we say that an online algorithm $ALG$ is $f(n)$-performant (resp. $f(n)$-competitive) if it guarantees a performance (resp. competitive ratio) of $f(n)$ for instances with $n$ requests. Hence:

$$\lambda_{ALG}^{\mathcal I_{n,T}}\geq f(n) ~~~ (\mathrm{resp.} ~ \gamma_{ALG}^{\mathcal I_{n,T}}\geq f(n)).$$


\end{definition}

Note that in Definition \ref{f(n)-comp}, the algorithm functions without knowing ahead of time the total number of requests that will appear. 

\begin{prop}
The functions $\lambda^{\mathcal I_{n,T}}$ and $\gamma^{\mathcal I_{n,T}}$ are non-decreasing with respect to $T$ and non-increasing with respect to $n$.
\end{prop}

\begin{proof}
Increasing $T$ restricts the strategy of the adversary without affecting the player's algorithm. Similarly, increasing $n$ increases the possibilities for the adversary. More precisely, if $T\leq T'$, then $\mathcal I_{n,T'}\subset \mathcal I_{n,T}$ and if $n\leq n'$, then $\mathcal I_{n,T}\subset \mathcal I_{n',T}$.
\end{proof}

\begin{remark} In order to find bounds to performances and competitive ratios, we may restrict our analysis by omitting inefficient algorithms. It is inefficient for the vehicle to slow down or to change directions when no request appears. Thus we will assume that the vehicle always moves at speed 1 and with a constant direction between the release times $t_i$ and $t_{i+1}$ of two successive requests.

\end{remark}

\begin{remark} Multiplying all the weights $\delta_i$ by a positive constant does not affect the game. Thus we will always suppose that the weight of the first request is $\delta_1=1$, except in the proof of Theorem \ref{large delay}.
\end{remark}

\subsection{Summary of results}

The optimal performance and competitive ratio are both decreasing with the number of requests $n$. For a fixed $n$, they are non-decreasing step functions with regards to the delay $T$. We distinguish three cases corresponding to small, large or medium delays. The following results were obtained for the case where $X$ is the segment $[-1,1]$.
 
 \begin{itemize}
     \item When there is no delay, there is a natural greedy algorithm for the vehicle which consists in always heading towards the request of greatest weight. This greedy algorithm is optimal, guaranteeing a performance and a competitive ratio of $\frac1n$. Small delays refer to values of $T$ small enough that this greedy algorithm remains optimal. This means $T<T_0=\frac 1{2^{n-3}+1}$ in terms of performance and  $T<T_1=\frac 1{2^{n-1}-2}$ in terms of competitive ratio.
 
 \item  With a delay $T$ greater or equal to 1, a more refined greedy algorithm, defined in Section \ref{T geq 1}, becomes optimal. The values of the optimal performances and optimal competitive ratios are then defined by a sequence $(\alpha_n)$ which is analysed in the appendix.

\item  The intermediate case is the most complex and remains mostly open. When the delay $T$ reaches the critical value $T_0$ (resp. $T_1$), the performance (resp. competitive ratio) increases by at least $\e=\frac 1{n(n-1)(n+3)}$. For $\frac12\leq T<1$, the performance is characterised by a sequence $(\beta_n)$, which is difficult to compute beyond the first few terms.

 \end{itemize}

 
 \begin{table}[ht]
    \centering
    \setcellgapes{3pt}
    \makegapedcells
    \begin{tabular}{|c|c|c|}
        
         \hline
         
         Delay &  Performance & Competitive Ratio\\
         \hline
         $T<T_1$ & $\frac 1n$ \hspace{3mm}\footnotesize{(**)} & $\frac 1n$ \hspace{7mm} \footnotesize{(**)}\\
         \hline
         $ T_1\leq T<T_0$ & $\frac 1n $ \hspace{3mm} \footnotesize{(**)} & $>\frac1n$\hspace{9mm} \\
         \hline
         $T_0\leq T<\frac 12$ & $\geq\frac 1n+\e$ \hspace{3mm} & $\geq\frac 1n+\e$ \hspace{5mm} \\
         \hline
         $\frac 12\leq T< 1$ & $\beta_n$ \hspace{3mm} \footnotesize{(*)}& $\geq \beta_n$ \hspace{5mm} \footnotesize{(*)}  \\
         \hline
         $1\leq T< 2-\frac 1{n-1}$ & $\alpha_{n-\lfloor \frac{1}{2-T}\rfloor}$ \footnotesize{(*)} & $\alpha_{n-\lfloor \frac{1}{2-T}\rfloor}$ \hspace{1mm} \footnotesize{(*)}\\
         \hline
         $2-\frac 1{n-1}\leq T$ & 1 \hspace{4mm} \footnotesize{(*)}& 1 \hspace{9mm} \footnotesize{(*)}\\
         \hline
    \end{tabular}
    \caption{Performances and Competitive Ratios with
     $T_0=\frac 1{2^{n-3}+1}$, $T_1=\frac 1{2^{n-1}-2}$, $\e=\frac 1{n(n-1)(n+3)}$. The results marked with an asterisk (resp. a double asterisk) can be generalised to centred geodesic metric spaces (resp. geodesic metric spaces) of diameter 2.} 
    \label{abcd}
\end{table}

These general results are summarised in Table \ref{abcd}. In sections 3 to 6, we study successively the performances and competitive ratios for no delay, large delays, small delays and medium delays on the segment $[-1,1]$. Section \ref{sec:small n} deals with values of $n\leq 4$, whence we are able to calculate explicitly the performance and competitive ratios for all the possible values of the delay. Results for $n=3$ and $n=4$ are summarised in Tables \ref{fig n=3 bis} and \ref{fig n=4 bis}. In Section \ref{sec:geodesic}, we generalise previous results to the case of centred geodesic metric spaces of diameter 2 (indicated with an asterisk in Table \ref{abcd}) or more generally to geodesic metric spaces of diameter 2 (indicated with a double asterisk). \emph{Centred} geodesic metric spaces of diameter $2$ are spaces included in $\overline B(O,1)$, where $\overline B(O,1)$ is the closed ball of radius $1$ centred in $O$.

\begin{table}[ht]
    \centering
    
    \setcellgapes{3pt}
    \makegapedcells

    \begin{tabular}{|c|c|c|}
        
         \hline
         Delay &  Performance & Competitive Ratio\\
         \hline
         $T<1/2$ & $1/3$ & $1/3$\\
         \hline
         $1/2\leq T< 1$ & $1/\varphi^2\approx 0.3820$ & $1/\varphi^2\approx 0.3820$\\
         \hline
         $1\leq T<1.5$ & $1/2$ & $1/2$\\
         \hline
         $1.5\leq T$ & 1 & 1\\
         \hline
    \end{tabular}
    \caption{Performances and Competitive Ratios for 3 requests ($n=3$)}
    \label{fig n=3 bis}
\end{table}

\begin{table}[ht]
    \centering
        \setcellgapes{3pt}
    \makegapedcells

    \begin{tabular}{|c|c|c|}
        
         \hline
         Delay &  Performance & Competitive Ratio\\
         \hline
         $T<1/6$ & \multirow{4}{*}{$1/4$} & $1/4$\\
         \cline{1-1} \cline{3-3} 
         $1/6\leq T<1/5 $ &  & $0.2578$\\
         \cline{1-1} \cline{3-3} 
         $1/5\leq T<1/4$ &  & $2-\sqrt{3}\approx 0.2679$\\
         \cline{1-1} \cline{3-3} 
         $1/4\leq T<1/3$ &  & $0.2803$\\
         \hline
         $1/3\leq T<1/2$ & $2-\sqrt{3}\approx 0.2679$ & $1-\sqrt 2/2\approx 0.2929$\\
         \hline
         $1/2\leq T< 1$ & $1-\sqrt{2}/2\approx 0.2929$ & $0.3177$\\
         \hline
         $1\leq T< 1.5$ & $1/\varphi^2\approx 0.3820$ & $1/\varphi^2\approx 0.3820$\\
         \hline
         $1.5\leq T< 5/3$ & $1/2$ & $1/2$\\
         \hline
         $5/3\leq T$ & 1 & 1\\
         \hline
    \end{tabular}
    \caption{Performances and Competitive Ratios for 4 requests ($n=4$)}
    \label{fig n=4 bis}
\end{table}

\section{\OPTiWinD$\;$with  no Delay}\label{T=0}

In this section, we will show that a greedy algorithm is optimal when there is no delay, and that no online algorithm guarantees a constant competitive ratio.

We consider a greedy algorithm for the vehicle, denoted by $GR_0$, in which the vehicle always moves towards the request of greatest weight. Thus, whenever a new request appears with a greater weight, the vehicle ignores his previous destination to go towards the new one.

\begin{lemma}\label{GR_0 perf}
The algorithm $GR_0$ is $\frac 1n$-performant; i.e. $\lambda_{GR_0}^{\mathcal I_{n,0}}=\frac 1n$.

\end{lemma}

\begin{proof}
By applying $GR_0$, the vehicle is guaranteed to serve the request of greatest weight. This is sufficient to ensure a performance at least equal to $1/n$.
\end{proof}

\noindent Note that it follows from Lemma \ref{GR_0 perf} that $GR_0$ is $\frac 1n$-competitive.


\begin{theorem}\label{GR_0 optimal}
For instances of $\mathcal I_{n,0}$, the algorithm $GR_0$ is optimal; i.e. $\gamma^{\mathcal I_{n,0}}=\frac 1n$
\end{theorem}

\begin{proof}
Let us consider an online algorithm $ALG$ for the vehicle. We will describe a strategy of the adversary which limits the competitive ratio of $ALG$ to $\frac 1n$. The adversary first releases a request $f_1=(-1,[1,3],1)$. If the vehicle does not go towards $-1$, the adversary will not release any more requests, and the performance will be $0$. Thus, we may assume that the vehicle will start moving towards $-1$ at time $1$. The strategy of the adversary is as follows: let $\t_1$ be the estimated time arrival (ETA) of the vehicle, i.e. $\t_1=2$. While the vehicle maintains his course and fewer than $n$ requests have been released, the adversary releases requests $$f_{1,k}=(1,[\t_1-\frac 1{4.3^k.n},\t_1-\frac 1{4.3^k.n}+2],1)\, ,\,\, k\geq 0.$$ If at time $\t_1-\frac 1{4.3^k.n}$, the vehicle changes directions to go towards $1$, the adversary sets $\t_2$ as the new ETA in 1. He then repeats the process, with $$f_{i,k}=((-1)^{i+1},[\t_i-\frac 1{4.3^k.n},\t_i-\frac 1{4.3^k.n}+2],1)$$ where $\t_i$ is the ETA in $(-1)^i$ after the vehicle's $i$-th change of directions, until $n$ requests have been released (see Figure \ref{Fig1} for an illustration). In this manner, a total of exactly $n$ requests will be released.  Also, when request $f_{i,k}$ is released, if $k>0$, it is already too late for the vehicle to reach $f_{i,k-1}$. Thus the vehicle will serve only one request. Note that $2i-\frac i{2n}\leq \t_i\leq 2i$. Hence, $2i-\frac 12\leq \t_i\leq 2i$. So $2i+1\in [\t_i-\frac 1{4.3^k.n},\t_i-\frac 1{2.3^k.n}+2]$. Hence, an optimal offline algorithm will serve all the requests by starting at $t=0$ towards $-1$ and reaching $(-1)^{i+1}$ at time $2i+1$.
Thus, no algorithm can guarantee a competitive ratio greater than $\frac 1n$. 
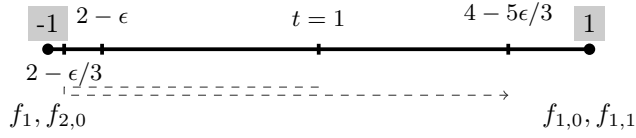
\begin{figure}
\begin{center}
\begin{tikzpicture}[scale=.72]
\draw [solid,line width=.5mm](0,0)--(10,0);
\draw [solid,line width=.5mm, color=black](5,-.1)--(5,.1) node[above=.1cm] {\small \color{black}$t=1$};
\fill(0,0) circle(3pt) node[above=.1cm, fill=black!20] {-1} node[below=.6cm] {\color{black}$f_1,f_{2,0}$};
\fill(10,0) circle(3pt) node[above=.1cm, fill=black!20] {1} node[below=.6cm] {\color{black}$f_{1,0}, f_{1,1}$};
\draw [solid,line width=.5mm, color=black](1,-.1)--(1,.1) node[above=.1cm] {\small \color{black}$2-\e$};
\draw [solid,line width=.5mm, color=black](.3,-.1)--(.3,.1) node[below=.1cm] {\small \color{black}$2-\e/3$};
\draw [solid,line width=.5mm, color=black](8.5,-.1)--(8.5,.1) node[above=.1cm] {\small \color{black}$4-5\e/3$};
\draw [color=black!80,->,dashed](5,-.7)--(.3,-.7)--(.3,-.85)--(8.5,-.85);
\end{tikzpicture}
\caption{Initial movements of the vehicle in the case where it changes direction toward $f_{1,1}$ at time $2-\e/3$. Using $\e=1/4n$, the figure displays the vehicle's position at each release time}.\label{Fig1}
\end{center}
\end{figure}
\end{proof}

\begin{remark}\label{ETA}
When a new request $f_i$ is released at time $t_i$, if $\t_i$ denotes the ETA for reaching $f_i$, then, given that the vehicle moves at speed $1$, the distance to $f_i$ at $t_i$ is $\t_i-t_i$.

Also, if the vehicle moves away from $f_i$ towards $f_k$ at time $t_i$, when the vehicle has passed half the distance to $f_k$, it is too late to serve $f_i$. Thus, if a request $f_{i+1}$ is released at the same end as $f_i$, it is impossible to serve both $f_i$ and $f_{i+1}$ if and only if $\t_k-t_{i+1}<\frac 12(\t_k-t_i)$.
\end{remark}


\begin{remark}
Note that Theorem \ref{GR_0 optimal} implies that there is no algorithm with a constant competitive ratio for $\TWOLF$ with zero delay.
\end{remark}


\section{\OPTiWinD$\;$with Large Delays}\label{T geq 1}

In this section, we will show that introducing a large delay significantly improves the performance and competitive ratio. The result shown in Theorem \ref{GR_0 optimal} indicates that the previous case is extremely unfavorable to the vehicle. However, this changes when we impose a delay between the release times of successive requests. With a delay sufficiently large, the strategy of the adversary used in the proof of Theorem \ref{GR_0 optimal} is no longer possible. Indeed, this strategy required that requests be released in rapid succession in $-1$ and $1$. 

Let us note that if $T \geq 2$ then $GR_0$ will serve all the requests. So, the only interesting case is when $T<2$. For $n\geq 1$, we define:

$$\alpha_n=\inf\limits_{\delta \in (\R_+^*)^{n}} \max \{\frac{\delta_1}{S_2}, \cdots , \frac{\delta_{i}}{S_{i+1}},\cdots, \frac{\delta_{n-1}}{S_{n}},\frac{\delta_{n}}{S_{n}}\}$$ where $S_i=\sum_{k=1}^i\d_k$. (See Appendix for the study of $(\alpha_n)$.)

\begin{theorem}\label{large delay}
For $1\leq T\leq 2-\frac 1{n-1}$, we have $\lambda^{\mathcal I_{n,T}}=\gamma^{\mathcal I_{n,T}}=\alpha_{n-\lfloor \frac{1}{2-T}\rfloor}$ and for $T\geq 2-\frac 1{n-1}$, $\lambda^{\mathcal I_{n,T}}=\gamma^{\mathcal I_{n,T}}=1$.
\end{theorem}

\begin{proof}
First, we will give a strategy of the adversary with a parameter $\e>0$ which limits the competitive ratio to $\alpha_{n-\lfloor \frac{1}{2-T}\rfloor}+O(\epsilon)$, when $\epsilon \rightarrow 0$. Then, we will give a greedy algorithm for the vehicle which guarantees a performance of $\alpha_{n-\lfloor \frac{1}{2-T}\rfloor}$.

\pagebreak

\begin{enumerate}
    \item Let $\epsilon >0$. The strategy of the adversary is defined as follows:
    \begin{enumerate}
        \item First, release a request $f_1=(-1,[1,3],\epsilon)$.
        \item If the vehicle does not serve $f_1$, do not release any more requests.
        \item Similarly, while the vehicle serves all the requests up to $f_{i-1}$ and $i\leq i_0=\lfloor \frac{1}{2-T}\big\rfloor$, release a request $$f_i=((-1)^i,[1+(i-1)T,3+(i-1)T],4^{i-1}\epsilon).$$
        \item If the vehicle serves $f_{i_0}$, release a request $$f_{i_0+1}=((-1)^{i_0+1},[1+i_0T,3+i_0T],1).$$
        \item Then, while the vehicle does not serve any more requests and $i\leq n$, release a request $$f_i=((-1)^i,[2i-2-\eta_i,2i-\eta_i],\delta_{i-i_0})$$ where $(\delta_i)$ realises $\alpha_{n-\lfloor \frac{1}{2-T}\rfloor}$ (see appendix) and $\eta_i=\frac{i-i_0}n-\frac 3{2n}$.
        \item When the vehicle serves a request $f_i$, $i>i_0$, do not release any more requests.
    \end{enumerate}
    
    Note that $\forall i\leq i_0+1$, $2i-1\in [1+(i-1)T,3+(i-1)T]$. Since $i_0\geq 1$, we have $\eta_i\leq \eta_n\leq 1-\frac5{2n}<1$. So, $\forall i\geq i_0+2$, $2i-1\in [2i-2-\eta_i,2i-\eta_i]$. Therefore, the optimal offline algorithm will serve all the requests by being in $(-1)^i$ at time $2i-1$. Also, the values of $\eta_i$ have been chosen so that for $i\geq i_0+2$, $t_i=\t_{i-1}-\frac1{2n}$.
    
    Let us now consider an online algorithm, $ALG$. When the adversary uses the above strategy, if $ALG$ does not allow the vehicle to serve $f_i$ for $i\leq i_0$, then the performance is $\frac{\sum\limits_{k=1}^{i-1} 4^{k-1}\epsilon}{\sum\limits_{k=1}^{i} 4^{k-1}\epsilon}=\frac{4^{i-1}-1}{4^i-1}<\frac 14$. It is shown in Proposition~\ref{alpha decreases} in the appendix that $\frac14<\alpha_{n-\lfloor \frac{1}{2-T}\rfloor}$.
    
    If $ALG$ allows the vehicle to serve the first $i_0$ requests, the vehicle will then reach at most one $f_i$ with $i>i_0$. Its performance will then be 
    \begin{eqnarray*}
    \frac{\delta_{i-i_0}+\sum\limits_{k=1}^{i_0} 4^{k-1}\epsilon}{\sum\limits_{k=1}^{i-i_0+1}\delta_k+\sum\limits_{k=1}^{i_0} 4^{k-1}\epsilon}&=&\frac{\delta_{i-i_0}}{S_{i-i_0+1}}+O(\epsilon)\\
    &=&\alpha_{n-\lfloor \frac{1}{2-T}\rfloor}+O(\epsilon).
    \end{eqnarray*}
    Therefore, $\gamma^{\mathcal I_{n,T}}\leq\alpha_{n-\lfloor \frac{1}{2-T}\rfloor}$.
    
    \pagebreak
    
    \item The greedy algorithm $GR_1$ is defined as follows:
    \begin{enumerate}
        \item When there is no request to be served, or none that can be reached in time, head towards $0$.
        \item While a single request that can be served is ongoing, head towards it.
        \item When a request $f_2=(x_2,[t_2,t_2+2],\delta_2)$ is released while another request $f_1=(x_1,[t_1,t_1+2],\delta_1)$ is still reachable, if $\frac{\delta_1}{\delta_1+\delta_2}\geq \alpha_{n-\lfloor \frac{1}{2-T}\rfloor}$, head towards $f_1$; else head towards $f_2$. From this point on, until a request is served, number the requests $f_1$ to $f_d$ and denote by $\delta_i$ the weight of $f_i$ and $S_i=\sum_{k=1}^i \delta_k$. When $f_{i+1}$ is released, if $\frac {\delta_i}{S_{i+1}}\geq \alpha_{n-\lfloor \frac{1}{2-T}\rfloor}$, head towards $f_i$; else, head towards $f_{i+1}$.
        
    \end{enumerate}
    
    We will divide the game into phases according to when the vehicle is either in cases a) and b) or in case c). 
    
    In the phases corresponding to a) and b), the vehicle will serve all the requests, and the performance over those phases of the game will be $1$.
    
    If at some point during the game, a sequence of requests $f_1,\ldots,f_d$ places the vehicle in case c), we will show that the performance will be greater or equal to $\alpha_{n-\lfloor \frac{1}{2-T}\rfloor}$ over this phase. With a delay of $T$ and the vehicle applying algorithm $GR_1$, the first time case c) may occur is for the $(i_0+1)$-th and $(i_0+2)$-th requests, where $i_0=\lfloor \frac{1}{2-T}\big\rfloor$. Hence, $d\leq n-i_0$. Note that when a request $f_i=(x_i,[t_i,t_i+2],\delta_i)$ is served at time $t$, since $T\geq 1$ and the time windows have length $2$, at most one request may have been released between $t_i$ and $t$. It follows that when $f_i$ is reached, the performance over that phase of the game is $\frac{\delta_i}{S_i}$ or $\frac{\delta_i}{S_{i+1}}$. By definition of $(\alpha_n)$, this performance is at least $\alpha_{n-i_0}$.
    
    Hence, in all phases of the game, the performance is at least $\alpha_{n-\lfloor \frac{1}{2-T}\rfloor}$. Therefore, $\lambda_{ALG}^{\mathcal I_{n,T}}\geq \alpha_{n-\lfloor \frac{1}{2-T}\rfloor}$.
    
    It follows from 1. and 2. that we have $\lambda^{\mathcal I_{n,T}}=\gamma^{\mathcal I_{n,T}}=\alpha_{n-\lfloor \frac{1}{2-T}\rfloor}$ for $1\leq T\leq 2-\frac 1{n-1}$.
    
    Hence, $\lambda^{\mathcal I_{n,2-\frac 1{n-1}}}=\gamma^{\mathcal I_{n,2-\frac 1{n-1}}}=\alpha_1=1$. Therefore, for $T\geq 2-\frac 1{n-1}$, we have $\lambda^{\mathcal I_{n,T}}=\gamma^{\mathcal I_{n,T}}=1$.
\end{enumerate}
\end{proof}

\section{\OPTiWinD$\;$with Small Delays}\label{small delays}

In this section, we will consider cases where $T$ is small enough that the results from Section \ref{T=0} still apply.

\begin{theorem}\label{frac 1n}
For $T<T_0=\frac 1{2^{n-3}+1}$, we have $\lambda^{\mathcal I_{n,T}}=\frac 1n$.
\end{theorem}

\begin{proof}
When the vehicle has to choose between two symmetrically placed requests of equal weights and cannot serve both, it is optimal to go towards the closest one. We will restrict our description of the strategy of the adversary by assuming that the vehicle follows this rule.

The strategy of the adversary is to release requests $$f_i=(-1,[t_i,t_i+2],1)$$ where $t_1=1$, $t_i=2-\frac{1-T}{2^{i-2}}+\frac{1-T}{2^{n-2}}-\frac T2$ for $2\leq i\leq n-1$, and $t_n=t_{n-1}+T$.
Note that $T<\frac 1{2^{n-3}+1}$ implies $T<\frac{1-T}{2^{n-3}}$. Hence, $t_{i+1}-t_i\geq T$ for all $1\leq i\leq n-1$. 

Let us show by induction that if the player follows the rule given at the beginning of this proof, starting from $t=1$, the vehicle will always move towards $-1$. This is true when the request $f_1$ appears. Let us assume that it moved towards $-1$ until $f_{i+1}$ appears. Thus, when $f_{i+1}$ is released, the vehicle is located at $1-t_{i+1}$. Hence, the distance between its current position and $1$ is $t_{i+1}$. Note that for $2\leq i \leq n-1$, we have $\frac 12(2-t_i)=\frac{1-T}{2^{i-1}}-\frac{1-T}{2^{n-1}}-\frac T4>2-t_{i+1}$. Hence, $2t_{i+1}>t_i+2$. Therefore, when $f_{i+1}$ is released, it is too late to serve $f_i$. Thus, the vehicle will continue towards $-1$. 

So, this strategy is well defined for optimal play. Also, the vehicle will only be able to serve a single request, and $n$ requests will be released.  Hence, the performance is at most $\frac 1n$. Since the algorithm $GR_0$ guarantees at least $\frac 1n$, we have an equality.
\end{proof}




\begin{theorem}\label{GREE0 superoptimal}
For $n\geq 3$ and $T<T_1=\frac 1{2^{n-1}-2}$, the algorithm $GR_0$ is optimal and $\gamma^{\mathcal I_{n,T}}=\frac 1n$.
\end{theorem}

\begin{proof}
Let us consider an online algorithm $ALG$ for the vehicle. We will describe a strategy of the adversary which limits the competitive ratio of $ALG$ to $\frac 1n$. The adversary first releases a request $f_1=(-1,[1,3],1)$. If the vehicle does not go towards $-1$, the adversary will not release any more requests, and the performance will be $0$. Thus, we may assume that the vehicle will start moving towards $-1$ at time $1$. The strategy of the adversary is as follows: 
\begin{enumerate}
\item let $\t_1$ be the estimated time arrival (ETA) of the vehicle in $-1$, i.e. $\t_1=2$. 
\item While the vehicle maintains its course and fewer than $n-1$ requests have been released, the adversary releases requests $$f_{1,k}=(1,[\t_1-2^{n-3-k}(T+2\epsilon)+\epsilon,\t_1-2^{n-3-k}(T+2\epsilon)+\epsilon+2],1)$$ for $0\leq k\leq n-3$ and some $0<\epsilon<T$. 
\item If at time $\t_1-2^{n-3-k}(T+2\epsilon)+\epsilon$, the vehicle chooses to change directions to go towards $1$, the adversary sets $\t_2$ as the new ETA in 1. 
\item He then repeats the process, with $$f_{i,k}=((-1)^{i+1},[\t_i-2^{n-2-i-k}(T+2\epsilon)+\epsilon,\t_i-2^{n-2-i-k}(T+2\epsilon)+\epsilon+2],1)$$ where $\t_i$ is the ETA in $(-1)^i$ after the vehicle's $i$-th change of directions, until $n-1$ requests have been released.  
\end{enumerate}
In this manner, a total of exactly $n-1$ requests will be released.  If the vehicle changed directions $j-1$ times in total, the $n$-th request will be $$f_n=((-1)^{j+1},[\t_j-\epsilon,\t_j-\epsilon+2],1).$$

We will show that a)  this strategy satisfies the condition regarding the delay $T$, b) the vehicle can only serve a single request and c) for $\epsilon$ small enough the optimal offline algorithm can serve all the requests.

\begin{enumerate}[label=\alph*)]
    \item Note that the delay between $f_{i,k}$ and $f_{i,k+1}$ is greater than $T$ (even $T+2\epsilon$), and the delay between $f_n$ and $f_{j,k}$, for the last value of $k$, is at least $T$. 
    \item When request $f_{i,k}$ is released, if $k>0$, it is already too late for the vehicle to reach $f_{i,k-1}$. Similarly, when $f_n$ is released, it is too late to serve $f_{j,k}$. Thus the vehicle will serve only one request. 
    \item If the vehicle changes direction when $f_{i,0}$ is released, the new ETA in $(-1)^{i+1}$ is   $\t_i+2-2^{n-1-i}(T+2\epsilon)+2\epsilon$. Hence, $$\t_{i+1}\geq \t_i+2-2^{n-1-i}(T+2\epsilon)+2\epsilon=\t_i+2-2^{n-1-i}T+O(\epsilon).$$ Therefore, we have $\t_i\geq 2i-(2^{n-1}-2^{n-i})T +O(\epsilon)$. Thus, the time window of $f_{i,k}$ closes at 
    $$\t_i-2^{n-2-i-k}T+2+O(\epsilon)\geq 2i +2-(2^{n-1}-2^{n-i}-2^{n-2-i-k})T+O(\epsilon).$$
    Since $2^{n-i}+2^{n-2-i-k}>2$, we have $(2^{n-1}-2^{n-i}-2^{n-2-i-k})T<1$. Hence, for $\epsilon$ small enough, $2i+1$ is in the time window of $f_{i,k}$. Similarly, the time window of $f_n$ closes at $$\t_j-\epsilon+2+O(\epsilon)\geq 2j +2-(2^{n-1}-2^{n-j})T+O(\epsilon)\geq 2j +2-(2^{n-1}-2)T+O(\epsilon).$$ Hence, for $\epsilon$ small enough, $2j+1$ is in the time window of $f_n$. Therefore, the optimal offline algorithm will serve all the requests by starting at $t=0$ towards $-1$ and reaching $(-1)^{i+1}$ at time $2i+1$.
\end{enumerate}
Thus no algorithm can obtain a competitive ratio greater than $\frac 1n$.
\end{proof}


\section {\OPTiWinD$\;$with Medium Delays}

\subsection{Tightness of the bounds $T_0$ and $T_1$}

In this section, we will show that when the delay $T$ is at least $T_0$ (resp. $T_1$), the optimal performance (resp. competitive ratio) is greater than $\frac1n$, meaning that the boundaries $T_0$ and $T_1$ are tight.

In order to show that the boundary $T_0$ is tight, we will use the following lemma, inspired by the proof of Theorem \ref{frac 1n}.

\begin{lemma}\label{2 fires}
Let $T\geq \frac 1 {2^{n-3}+1}$. Starting at time 1, if the vehicle, initially at 0, always moves towards the request $f_1=(-1,[1,3],\delta_1)$, and if $n-1$ other requests are released before time 2, there will be a moment when the vehicle has the possibility of serving two requests.
\end{lemma}

\begin{proof}
Since the vehicle is moving towards $-1$, we may assume that the other $n-1$ requests are all located in $1$; thus we will use the same notations as in the proof of Theorem \ref{frac 1n}. The earliest possible release time for $f_2$ is $1+T$. For $2\leq i\leq n-1$, 
\begin{itemize}
    \item if $2-t_{i+1}\geq\frac 12(2-t_i)$, it follows from Remark \ref{ETA} that it is possible for the vehicle to serve both $f_i$ and $f_{i+1}$. 
    \item if $2-t_{i+1}<\frac 12(2-t_i)$, then $t_{n-1}>2-T$. Thus, $t_n>2$; and it is possible to serve both $f_1$ and $f_n$.  
\end{itemize}
\end{proof}

\begin{theorem}\label{Al1}
For $T_0=\frac 1 {2^{n-3}+1}$, we have $\lambda^{\mathcal I_{n,T_0}}\geq\frac 1n +\epsilon$, where $\epsilon=\frac 1{n(n-1)(n+3)}$.
\end{theorem}

\begin{proof}
We define the algorithm $Al1$ as follows:
\begin{enumerate}[label=(\alph*)]
    \item When the first request $f_1$ is released, head towards it. We may assume that $f_1=(-1,[1,3],1)$.
    \item While no weight is greater than $\k=\frac{1+n\epsilon}{1-n(n-1)\epsilon}$ and the vehicle does not have the possibility of serving two requests, keep going towards $f_1$.
    \item If a request $f_i$ appears with weight  $\delta_i> \k$, head towards it and behave like $GR_0$ from this point onwards.
    \item If at some point the vehicle has the possibility of serving two requests, combine their weights and behave like $GR_0$.
\end{enumerate}

We will show that $\lambda_{Al1}^{\mathcal I_{n,T_0}}\geq \frac 1n+\epsilon$.

\vspace{.2cm}
\noindent - Case (b): If no weight is greater than $\k$ and the vehicle never has the possibility of serving two requests, it will reach $f_1$ at time $2$. It follows from Lemma \ref{2 fires} that at most $n-1$ requests have been released at this point. Therefore, the vehicle's performance is at least $\frac 1{\k(n-1)}$. Note that with $\epsilon=\frac 1{n(n-1)(n+3)}$, we have  $$\frac 1{\k(n-1)}=\frac{1-n(n-1)\epsilon}{1+n\epsilon}\frac 1{n-1}= \frac1n\frac{n^2+2n}{n^2+2n-2}>\frac 1n\frac{n^2+2n-2}{n^2+2n-3}=\frac 1n+\epsilon.$$

\vspace{.2cm}
\noindent- Case (c): If $\delta_1=1$ and $\delta_i> \frac{1+n\epsilon}{1-n(n-1)\epsilon}$, we cannot have both $\max (\delta_i)\leq (\frac 1n+\epsilon)S_n$ and $\min(\delta_i)\geq (\frac 1n-(n-1)\epsilon)S_n$. And since $\max (\delta_i)\leq (\frac 1n+\epsilon)S_n$ implies $\min(\delta_i)\geq (\frac 1n-(n-1)\epsilon)S_n$, we have: $\frac {\max(\delta_i)}{S_n}>\frac 1n+\epsilon$.

\vspace{.2cm}
\noindent - Case (d): Let us assume that, at time $t_j$, the vehicle becomes able to serve two requests of weights $\delta_i$ and $\delta_j$. Let $\Delta=\{ 1,\ldots, n+1\}-\{i,j\}$ and let $\delta_{n+1}=\delta_i+\delta_j$. Let $J=\{1\leq k\leq j|f_k ~\mathrm{is ~no ~longer ~reachable ~at ~time } ~t_j\}$ and let $\delta=\max_{k\in \Delta-J} \delta_k$. We define $S_\Delta=\sum_{k\in\Delta} \delta_k$ and $S_J$ similarly. Note that when the possibility of serving two requests appears at time $t_j$, the request $f_1$ is still reachable. Hence, $1\notin J$ and $\delta\geq 1$. Also, since we are in case (d), $\forall k\in J$, $\delta_k\leq \k$. By behaving like $GR_0$, the vehicle is sure to serve a request of weight at least $\delta$. As $\k >1$, we have: $$S_\Delta\leq(n-1-|J|)\delta\leq\k\delta(n-1-|J|) ~\text{and}~ S_J\leq \k|J|\leq \k\delta|J|.$$ 
Hence, the vehicle's optimal performance is at least equal to: $$\frac\delta{S_n}=\frac\delta{S_\Delta+S_J}\geq\frac 1{\k(n-1)}>\frac 1n+\epsilon.$$
\end{proof}

\begin{corollary}
We have $\lambda^{\mathcal I_{n,T}}=\frac 1n$ if and only if $T<\frac 1{2^{n-3}+1}$.
\end{corollary}

\begin{proof}
This follows directly from Theorem \ref{frac 1n}, Theorem \ref{Al1} and the fact that $\lambda^{\mathcal I_{n,T}}$ increases with $T$.
\end{proof}

In order to prove that the bound given in Theorem \ref{GREE0 superoptimal} is tight, we define the algorithm $Al2$ as follows:
\begin{enumerate}[label=(\alph*)]
    \item When $f_1$ is released at $t_1=1$ in $-1$, go towards $-1$.
    \item When $f_2$ is released at $t_2$ in $1$, if $t_2>2-2^{n-3}T_1$, keep going towards $-1$ either until two requests can be served by changing direction, or until $f_1$ is reached; else go towards $1$ and set $\t_2$ as the ETA in $1$.
    \item Similarly, while the vehicle has changed direction with each new request, when $f_i$ is released at $t_i$ in $(-1)^i$, if $t_i>\t_{i-1} - 2^{n-1-i}T_1$, keep going towards $(-1)^{i-1}$ either until two requests can be reached by changing direction, or until $f_{i-1}$ is reached; else change direction towards $(-1)^i$ and set $\t_i$ as the ETA in $(-1)^i$.
\end{enumerate}

\begin{lemma}\label{Al2}
For $n\geq 3$ and given a delay $T\geq T_1=\frac 1 {2^{n-1}-2}$, the algorithm $Al2$ guarantees that if $n$ requests are released, either the vehicle can serve at least two requests or no optimal offline algorithm can serve all requests. 
\end{lemma}

\begin{proof}
We may assume that the first request released by the adversary is $f_1=(-1,[1,3],1)$ and the vehicle applies $Al2$. Assuming it exists, let $i$ denote the smallest index for which $t_i>\t_{i-1} - 2^{n-1-i}T$. The vehicle will keep going towards $(-1)^{i-1}$. For $i \leq k\leq n-1$, if $\t_{i-1}-t_{k+1}\geq \frac 12(\t_{i-1}-t_k)$, then it follows from Remark \ref{ETA} that the vehicle can serve both $f_k$ and $f_{k+1}$. However, if $\t_{i-1}-t_{k+1}< \frac 12(\t_{i-1}-t_k)$, for all $i\leq k\leq n-1$, then $\t_{i-1}-t_{n-1}<T$ so the $n$-th request cannot be released before $f_{i-1}$ is reached.

Since $\forall i\geq 2$, $\t_i=2t_i-\t_{i-1}+2$, if $\forall 2\leq i\leq n-1$, $t_i\leq \t_{i-1} - 2^{n-1-i}T$, then $$\t_{n-1}\leq \t_1+2(n-2)-\sum_{i=1}^{n-2} 2^iT\leq \t_1+2(n-2)-\sum_{i=1}^{n-2} 2^iT_1= 2n-3.$$ So if $t_n< \t_{n-1}$, then $2n-1\notin [t_n,t_n+2]$. Hence, either the vehicle can serve two requests, or the optimal offline algorithm cannot serve all the requests.
\end{proof}

\begin{theorem}\label{T1}
For $T_1=\frac 1 {2^{n-1}-2}$, we have $\gamma^{\mathcal I_{n,T_1}}>\frac 1n$.
\end{theorem}


\begin{proof}
Let $0<\e<\frac1{n^2}$, $\k=\frac{1+n\epsilon}{1-n(n-1)\epsilon}$ and let the sequence  $(\o_i)$ be defined by $\o_1=1$ and $\o_{i+1}=\o_i(1-n^2\e)-(i-2)(\k-\o_i)$.
We define the algorithm $Al3$ as follows:
\begin{enumerate}[label=(\alph*)]
    \item When $f_1=(-1,[1,3],1)$ is released and while each request $f_i$ has a weight satisfying $\o_i\leq \delta_i\leq\k$ and no two requests can be served, apply $Al2$.
    \item If the request $f_i$ has a weight $\delta_i>\k$ or $\delta_i<\o_i$, apply $GR_0$.
    \item If two requests can be served, combine them and apply $GR_0$.
    
\end{enumerate}

We will show that for $\e$ small enough, $Al3$ guarantees a competitive-ratio of $\frac 1n+\e$.

We can show by induction that $\o_i\leq1$ and $\o_i=1+O(\e)$. Since $\k>1\geq \o_i$, we have $(\o_i)$ is decreasing. Note also that $\k=1+O(\e)$, so $\frac{\o_i}\k=1+O(\e)$.

- Case (a): Assume that for all $i$, $\o_i\leq \d_i\leq \k$ and no two requests can ever be served. Then, the vehicle applies $Al2$ until the end. It follows from Lemma \ref{Al2} that the optimal offline algorithm will not be able to serve $n$ requests. Thus the competitive ratio in this case is at least $\frac{\o_n}{(n-1)\k}=\frac 1{n-1}+O(\e)$. Hence, for $\e$ small enough, this competitive ratio is greater than $\frac 1n+\e$.

- Case (b.1): Assume now that the vehicle switches from $Al2$ to $GR_0$ when a request $f_i$ is released with $\d_i>\k$. This case is identical to Case (c) in the proof of Theorem \ref{Al1}. Thus, the performance is at least $\frac1n+\e$.

- Case (b.2): Assume now that the vehicle switches from $Al2$ to $GR_0$ when a request $f_i$ has a weight $\d_i<\o_i$. Let $j$ denote the index of the request towards which the vehicle was heading when $f_i$ was released. Let $\d=\max \{\d_k,k>i\}\cup\{\d_j\}$. The vehicle will at least serve a request of weight $\d$, with $\d\geq\o_{i-1}$ and $\d_i<\o_i$. The sum of the weights is at most $(n-i+1)\d+\d_i+(i-2)\k$. So the competitive ratio is at least 
\begin{eqnarray*}
\lefteqn{
\frac \d{(n-i+1)\d+\d_i+(i-2)\k} \geq \frac{\o_{i-1}}{(n-i+1)\o_{i-1}+\o_i+(i-2)\k}}\hspace{3cm}\\
&\geq& \frac{\o_{i-1}}{(n-i+1)\o_{i-1}+\o_{i-1}(1-n^2\e)+(i-2)\o_{i-1}}\\
&\geq&\frac 1{n-n^2\e}\\
&\geq& \frac1n+\e.
\end{eqnarray*}

- Case (c): Assume now that the vehicle switches from $Al2$ to $GR_0$ when two requests can be served. This case is similar to Case (d) of the proof of Theorem \ref{Al1}. The only difference is that we no longer have $\d\geq 1$. Using the same notations, we have instead $\d\geq \o_{j-1}$. Thus the performance is at least $\frac {\o_{j-1}}{\k(n-1)}$. For $\e$ small enough, this is greater than $\frac1n+\e$.

Hence, for $\e$ small enough, $Al3$ guarantees a competitive ratio of $\frac1n+\e$.
\end{proof}

We close this section by showing that the optimal performance for $\frac 12\leq T<1$ can be expressed using an induction formula. In order to solve the general case for $\frac 12\leq T<1$, we introduce a variant problem in which the initial state is modified. Thus we define the \emph{optimal} weighted performance $\lambda^{\mathcal I_{n,T},\delta_0}$ as the optimal performance obtained by an algorithm in the case where the vehicle's starting position at $t=0$ is somewhere in $]T-1,1-T[$, requests with a total weight of $\delta_0\geq 0$ have already been missed and the next request to appear is of weight $1$. We will show that this is well defined as the optimal weighted performance does not depend on where in $]T-1,1-T[$ the vehicle starts.

\begin{theorem}\label{weighted performance}
For $\frac 12\leq T<1$, the optimal weighted performance $\lambda^{\mathcal I_{n,T},\delta_0}$ satisfies the following:
\begin{itemize}
    \item $\lambda^{\mathcal I_{1,T},\delta_0}=\frac 1{1+\delta_0}\;\;$ and $\;\;\lambda^{\mathcal I_{2,T},\delta_0}=\frac 1{2+\delta_0}$,
    \item $\forall n\geq 3$, 
\begin{eqnarray*}
    \lambda^{\mathcal I_{n,T},\delta_0}&=&\inf\limits_{\delta_2\geq 0} \max \{\frac 1{\delta_0+1+\delta_2},\min\{\lambda^{\mathcal I_{n-2,T}, \frac{\delta_0+1+\delta_2}{\delta_2}}, \\& &\lambda^{\mathcal I_{n-1,T}, \frac{\delta_0+1}{\delta_2}}\} \}.
    \end{eqnarray*}
\end{itemize}    
\end{theorem}

\begin{proof}
If $n=1$, the vehicle will reach the request of weight $1$, so its performance will be $\frac 1 {1+\delta_0}$. If $n=2$, the vehicle will always be able to reach one of the two requests, and the worst case is when both requests are released with delay $T$ at opposite ends and with equal weights. In that case, the vehicle's performance is $\frac 1{2+\delta_0}$. Let us now consider $n\geq 3$. We may assume that the first request is $f_1=(-1,[0,2],1)$ and that the vehicle will move towards it. After a delay of at least $T$ and before $f_1$ is served, the adversary will release a request $f_2=(1,[t_2,t_2+2],\delta_2)$. The vehicle will thus have two options. If it continues towards the first request, it will serve $f_1$ before a third request can be released and its performance will be $\frac 1{\delta_0+1+\delta_2}$. If it changes direction towards $f_2$, it will have returned to the interval $]T-1,1-T[$ at time $t_2+T$. While the vehicle remains in this interval, the adversary may release a request $f_3=(-1,[t_3,t_3+2],\delta_3)$, and the vehicle will have to choose between $f_2$ and $f_3$. In this case, it follows from the symmetry of that situation that it is optimal for the adversary to select $\delta_3=\delta_2$. If we divide all the weights by $\delta_2$, the situation is identical to that of the problem defined above, with $n$ replaced by $n-2$ and $\delta_0$ replaced by $\frac{\delta_0+1+\delta_2}{\delta_2}$. While the vehicle is in the interval $]T-1,1-T[$, the adversary may choose not to release a request, and similarly, in this case, $n$ is replaced by $n-1$ and $\d_0$ is replaced by $\frac{\delta_0+1}{\delta_2}$. Hence,    
\begin{eqnarray*}
    \lambda^{\mathcal I_{n,T},\delta_0}&=&\inf\limits_{\delta_2\geq 0} \max \{\frac 1{\delta_0+1+\delta_2},\min\{\lambda^{\mathcal I_{n-2,T}, \frac{\delta_0+1+\delta_2}{\delta_2}}, \\& &\lambda^{\mathcal I_{n-1,T}, \frac{\delta_0+1}{\delta_2}}\} \}.
\end{eqnarray*}
\end{proof}

\begin{corollary}\label{weightless induction}
For $\frac 12\leq T<1$, we have $\lambda^{\mathcal I_{n,T}}=\beta_n$, where $\beta_1=1$, $\beta_2=\frac12 $ and 
$ \forall n\geq 3$,$$\beta_n=\inf\limits_{\delta_2\geq 0} \max \{\frac 1{1+\delta_2},\min\{\lambda^{\mathcal I_{n-2,T}, \frac{1+\delta_2}{\delta_2}},\lambda^{\mathcal I_{n-1,T}, \frac{1}{\delta_2}}\} \}.$$
\end{corollary}

\begin{proof}
It follows from the definitions that the optimal performance $\lambda^{\mathcal I_{n,T}}$ is equal to the optimal weighted performance $\lambda^{\mathcal I_{n,T},0}$ with zero requests missed.
\end{proof}

\begin{conj}
The formulas in Theorem \ref{weighted performance} and Corollary \ref{weightless induction} can be simplified as follows: for all $\frac 12\leq T<1$ and $n\geq 3$,
$$ \lambda^{\mathcal I_{n,T},\delta_0}=\inf\limits_{\delta_2\geq 0} \max \{\frac 1{\delta_0+1+\delta_2},\lambda^{\mathcal I_{n-2,T}, \frac{\delta_0+1+\delta_2}{\delta_2}}\} $$ 
$$\text{and}~\;\;\; \lambda^{\mathcal I_{n,T}}=\inf\limits_{\delta_2\geq 0} \max \{\frac 1{1+\delta_2},\lambda^{\mathcal I_{n-2,T},\frac{1+\delta_2}{\delta_2}}\}.$$
\end{conj}

In the proof of Proposition \ref{small n medium T} below, we verify that for $n=3$ or $4$, we have $\lambda^{\mathcal I_{n-2,T}, \frac{\delta_0+1+\delta_2}{\delta_2}}\leq \lambda^{\mathcal I_{n-1,T}, \frac{\delta_0+1}{\delta_2}}$. Using Proposition \ref{small n medium T}, we can easily verify that it is also true for $n=5$. So the conjecture is true for small values of $n$. The conjecture states that it is inefficient for the adversary to pass on the possibility of releasing a request when the vehicle comes back near $0$.

\bigskip
\noindent Table \ref{abc} summarises the results obtained so far:

\begin{table}[ht]
    \centering
    \setcellgapes{3pt}
    \makegapedcells
    \begin{tabular}{|c|c|c|}
        
         \hline
         Delay &  Performance & Competitive Ratio\\
         \hline
         $T<T_1$ & $\frac 1n$ & $\frac 1n$\\
         \hline
         $ T_1\leq T<T_0$ & $\frac 1n$ & $>\frac1n$\\
         \hline
         $T_0\leq T<\frac 12$ & $\geq\frac 1n+\e$ & $\geq\frac 1n+\e$\\
         \hline
         $\frac 12\leq T< 1$ & $\beta_n$ & $\geq \beta_n$  \\
         \hline
         $1\leq T< 2-\frac 1{n-1}$ & $\alpha_{n-\lfloor \frac{1}{2-T}\rfloor}$ & $\alpha_{n-\lfloor \frac{1}{2-T}\rfloor}$\\
         \hline
         $2-\frac 1{n-1}\leq T$ & 1 & 1\\
         \hline
    \end{tabular}
    \caption{Performances and Competitive Ratios with
     $T_0=\frac 1{2^{n-3}+1}, T_1=\frac 1{2^{n-1}-2}$, $\e~=~\frac 1{n(n-1)(n+3)}$.}
    \label{abc}
\end{table}


The performance remains unknown for $T_0\leq T<\frac 12$, and the competitive ratio for $T_1\leq T< 1$.

\section{\OPTiWinD$\;$with a total number of requests at most $4$}\label{sec:small n}

In this section, we will give a complete description of what happens when the number of requests is at most $4$.

When there is only one request, the vehicle can always reach it, so the performance and competitive ratio are equal to 1.

When the number of requests is at most two,  the performance and competitive ratio are equal to $\frac 12$ if $T<1$ and $1$ otherwise.

For $n=3$, the cases $T< \frac 12$ and $T\geq 1$ are covered by Theorem \ref{GREE0 superoptimal} and Theorem \ref{large delay}, respectively. The case $\frac 12\leq T<1$ is given as an induction formula in Theorem \ref{weighted performance}. In the following proposition, we calculate the induction formula for $n=3$ and $n=4$.


\begin{prop}\label{small n medium T}
For $\frac 12\leq T<1$, we have $\lambda^{\mathcal I_{3,T},\delta_0}=\frac 2{3+\delta_0+\sqrt{\delta_0^2+2\delta_0+5}}$ and $\lambda^{\mathcal I_{4,T},\delta_0}=\frac 2{4+\delta_0+\sqrt{\delta_0^2+8}}$. 
\end{prop}

\begin{proof}
Using the formulas given in Theorem \ref{weighted performance}, we obtain:

$$\lambda^{\mathcal I_{3,T},\delta_0}=\inf\limits_{\delta_2\geq 0} \max\big\{\frac1{\delta_0+1+\delta_2},\frac 1{1+\frac{\delta_0+1+\delta_2}{\delta_2}}\big\}.$$

In the $\inf\max$ above, the first term is decreasing while the second is increasing in $\delta_2$. Hence, the $\inf\max$ is reached when there is equality.

\begin{eqnarray*}
\lefteqn{\delta_0+1+\delta_2=1+\frac{\delta_0+1+\delta_2}{\delta_2}}~~~~~~~~~~~~~~~~~~~~~~~~~~~~\\
&\Leftrightarrow& \delta_2^2+(\delta_0-1)\delta_2-(\delta_0+1)=0\\
&\Leftrightarrow& \delta_2=\frac{1-\delta_0+\sqrt{\delta_0^2+4\delta_0+5}}2
\end{eqnarray*}
Hence, $\lambda^{\mathcal I_{3,T},\delta_0}=\frac 2{3+\delta_0+\sqrt{\delta_0^2+4\delta_0+5}}$.

\vspace{.2cm}
It follows that for $\d'_0=\frac{\delta_0+1}{\delta_2}$, we have
\begin{eqnarray*}
\lambda^{\mathcal I_{3,T},\delta'_0}&=&\frac2{3+\d'_0+\sqrt{{\d'_0}^2+4\d'_0+5}}\\
&<&\frac2{3+\d'_0+\sqrt{(\d'_0+3)^2}}=\frac1{3+\d'_0}=\lambda^{\mathcal I_{2,T},1+\delta'_0}.
\end{eqnarray*}
Hence, $\lambda^{\mathcal I_{4,T},\delta_0}=\inf\limits_{\delta_2\geq 0} \max\big\{\frac1{\delta_0+1+\delta_2},\frac 1{2+\frac{\delta_0+1+\delta_2}{\delta_2}}\big\}.$

Thus, 
\begin{eqnarray*}
\lefteqn{\delta_0+1+\delta_2=2+\frac{\delta_0+1+\delta_2}{\delta_2}} ~~~~~~~~~~~~~~~~~~~~~~~~\\
&\Leftrightarrow& \delta_2^2+(\delta_0-2)\delta_2-(\delta_0+1)=0\\
&\Leftrightarrow& \delta_2=\frac{2-\delta_0+\sqrt{\delta_0^2+8}}2.
\end{eqnarray*}
Hence, $\lambda^{\mathcal I_{4,T},\delta_0}=\frac 2{4+\delta_0+\sqrt{\delta_0^2+8}}$.
\end{proof}

\begin{corollary}\label{3&4}
For $\frac 12\leq T<1$, we have $\lambda^{\mathcal I_{3,T}}=\frac 1{\varphi^2}$, where $\varphi$ is the golden ratio, and $\lambda^{\mathcal I_{4,T}}=1-\frac{\sqrt 2}2$.
\end{corollary}

\begin{prop}\label{golden}
For $\frac 12\leq T<1$, we have $\gamma^{\mathcal I_{3,T}}=\frac 1{\varphi^2}$, where $\varphi$ is the golden ratio.
\end{prop}

\begin{proof}
For $\frac 12\leq T<1$ and $n=3$, the adversary can limit the competitive ratio by using the following strategy: 
\begin{enumerate}[label=(\alph*)]
    \item First, release a request $f_1=(-1,[1,3],1)$.
    \item Assuming the vehicle moves towards $f_1$, release  $f_2=(1,[2-\e,4-\e],\varphi)$, with $\e<\min\{\frac13,1-T\}$.
    \item If the vehicle keeps going towards $f_1$, do not release any more requests.
    \item Else, if the vehicle goes towards $f_2$, release  $f_3=(-1,[4-3\e,6-3\e],\varphi)$. (see Figure \ref{Fig2})
\end{enumerate}


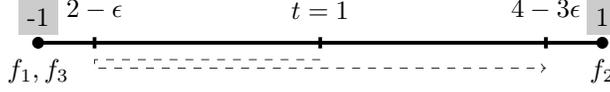
\begin{figure}
\centering
\begin{tikzpicture}[black,scale=.75]
\draw [solid,line width=.5mm](0,0)--(10,0);
\draw [solid,line width=.5mm, color=black](5,-.1)--(5,.1) node[above=.1cm] {\color{black}$t=1$};
\fill(0,0) circle(3pt) node[above=.1cm,fill=black!20] {-1} node[below=.1cm] {\color{black}$f_1,f_3$};
\fill(10,0) circle(3pt) node[above=.1cm,fill=black!20] {1} node[below=.1cm] {\color{black}$f_2$};
\draw [solid,line width=.5mm, color=black](1,-.1)--(1,.1) node[above=.1cm] {\color{black}$2-\e$};
\draw [solid,line width=.5mm, color=black](9,-.1)--(9,.1) node[above=.1cm] {\color{black}$4-3\e$};
\draw [color=black!80,->,dashed](5,-.3)--(1,-.3)--(1,-.45)--(9,-.45);
\end{tikzpicture}
\caption{vehicle's itinerary for $n=3$ and $\frac 12\leq T<1$ in Proposition \ref{golden}.}\label{Fig2}
\end{figure}


Since $2i-1\in [t_i,t_i+2]$, the offline algorithm can serve all the requests. If the vehicle chooses to keep going towards $f_1$ when $f_2$ is released, the competitive ratio will be $\frac 1{1+\varphi}$. Else, if it changes directions, the competitive ratio will be $\frac{\varphi}{1+2\varphi}=\frac 1{1+\varphi}=\frac1{\varphi^2}$. 
\end{proof}

In Table \ref{fig n=3}, we summarise the results obtained for $n=3$.

\begin{table}[ht]
    \centering
    \setcellgapes{3pt}
    \makegapedcells
    
    \begin{tabular}{|c|c|c|}
        
         \hline
         Delay &  Performance & Competitive Ratio\\
         \hline
         $T<1/2$ & $1/3$ & $1/3$\\
         \hline
         $1/2\leq T< 1$ & $1/\varphi^2\approx 0.3820$ & $1/\varphi^2\approx 0.3820$\\
         \hline
         $1\leq T<1.5$ & $1/2$ & $ 1/2$\\
         \hline
         $1.5\leq T$ & 1 & 1\\
         \hline
    \end{tabular}
    \caption{Performances and Competitive Ratios for $n=3$}
    \label{fig n=3}
\end{table}

For $n=4$, the cases $T< \frac 16$ and $T\geq 1$ are covered by Theorem \ref{GREE0 superoptimal} and Theorem \ref{large delay}, respectively. In the case where $\frac 12\leq T<1$, the performance is given in Corollary \ref{3&4}.

Note that for $n=4$ and $T<\frac 16$, the game  can be represented as a decision tree shown in Figure \ref{decision tree}. Assuming $\d_1=1$, $\d_4=\d_1$, $\d'_4=\d_3$, $\d''_4=\d'_3$ and $\d'''_4=\d_2$, the optimal performance is $$\lambda^{\mathcal I_{4,T}}\hspace{-.15cm}=\hspace{-.05cm}\min\limits_{\d_2}\max\{
\min\limits_{\d_3}\max\{\frac{\d_1}{S_4},\frac{\d_3}{S'_4}\}, \min\limits_{\d'_3}\max\{\frac{\d'_3}{S''_4},\frac{\d_2}{S'''_4}\}\}.$$

\begin{figure}
\centering
\begin{tikzpicture}[black,scale=.6]
\node{$f_1$}[level distance=3cm,grow=right]
    child{node{$f_2$}[sibling distance=2.5cm] 
        child{node{$f_3$}[sibling distance=1.2cm]
            child{node{$f_4$} node[right=.3cm,fill=black!20]{~1}
            node[right=1cm]{$\frac{\d_1}{S_4},\d_1=\d_4$}
            edge from parent node[below,sloped]{\small \textit{toward -1}}}
            child{node{$f'_4$}node[right=.3cm,fill=black!20]{-1}
            node[right=1cm]{$\frac{\d_3}{S'_4},\d_3=\d'_4$}
            edge from parent node[above,sloped]{\small \textit{toward 1}}}
            node[below=.3cm,fill=black!20]{~1}
            edge from parent node[below,sloped]{\small \textit{toward -1}}
        } 
        child{node{$f'_3$}[sibling distance=1.2cm]
            child{node{$f''_4$}node[right=.3cm,fill=black!20]{~1}
            node[right=1cm]{$\frac{\d'_3}{S''_4},\d'_3=\d''_4$}
            edge from parent node[below,sloped]{\small \textit{toward -1}}}
            child{node{$f'''_4$}node[right=.3cm,fill=black!20]{-1}
            node[right=1cm]{$\frac{\d_2}{S'''_4},\d_2=\d'''_4$}
            edge from parent node[above,sloped]{\small \textit{toward 1}}}
            node[below=.3cm,fill=black!20]{-1}
            edge from parent node[above,sloped]{\small \textit{toward 1}}
        } node[below=.3cm,fill=black!20]{~1} 
        edge from parent node[above,sloped]{\small \textit{toward -1}}
        } node[below=.3cm,fill=black!20] {-1};
\end{tikzpicture}
\caption{Decision tree for the vehicle's movements, $T<1/6$ and at most $4$ requests. The locations of the requests are given in the grey rectangles. The value of the competitive ratio and the condition on the weight of the fourth request relative to each branch are indicated on the right.}\label{decision tree}
\end{figure}
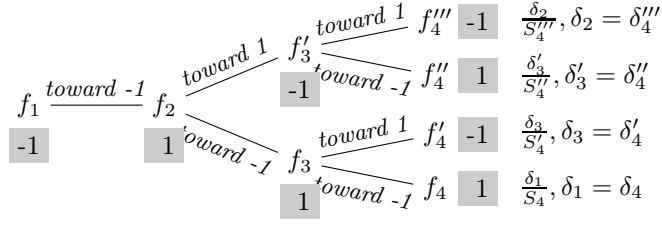

\begin{prop}
For $\frac 13\leq T<\frac12$, we have $\lambda^{\mathcal I_{4,T}}=2-\sqrt 3$.
\end{prop}

\begin{proof}
For $\frac 13\leq T<\frac12$, when the vehicle keeps moving towards $-1$, the adversary cannot release the fourth request before $f_1$ is served. This corresponds to replacing $\d_4=\d_1$ in the case shown in Figure \ref{decision tree} with $\d_4=0$. 
By setting $\frac{\d'_3}{S''_4}=\frac{\d_2}{S'''_4}$, we obtain $\d'_3=\d_2$. Setting $\frac{\d_1}{S_4}=\frac{\d_3}{S'_4}$ with $\d_4=0$ gives $\d_3=\frac{1-\d_2}2+\frac12\sqrt{\d_2^2+2\d_2+5}$. Then, setting $\frac{\d_1}{S_4}=\frac{\d_2}{S'''_4}$, we deduce that $\d_2$ is a root of $2X^3-3X-1$. Hence, $\d_2=\frac{1+\sqrt 3}2$, and $\lambda^{\mathcal I_{4,T}}=2-\sqrt 3$.

The adversary can limit the performance to $2-\sqrt 3$ by using the following strategy with the optimal weights found above and $\e=\frac12 -T$:
\begin{enumerate}[label=(\alph*)]
    \item First, release a request $f_1=(-1,[1,3],1)$.
    \item  Assuming the vehicle moves towards $f_1$, release $f_2=(1,[2-T-\e,4-T-\e],\d_2)$.
    \item If the vehicle keeps going towards $-1$, release $f_3=(1,[2-\e,4-\e],\d_3)$. If the vehicle keeps going towards $-1$, do not release a fourth request; else, if it changes directions, release a request $f'_4=(-1,[4-3\e,6-3\e],\d'_4)$.
    \item Else, if the vehicle changes directions, release a request $f'_3=(-1,[4-3T-3\e,6-3T-3\e],\d'_3)$. If it changes directions again, release a request $f''_4=(1,[6-4T-5\e,8-4T-5\e],\d''_4)$; else, release $f'''_4=(-1,[4-2T-3\e,6-2T-3\e],\d'''_4)$.
\end{enumerate}

This strategy satisfies the condition relative to the delays and prevents the vehicle from serving two requests. Hence, $\lambda^{\mathcal I_{4,T}}=2-\sqrt 3$.
\end{proof}

In order to avoid a lengthy case study, the remaining competitive ratios for $n=4$ will be listed without detailed proof; but we will give a general idea of how they can be obtained.

\begin{claims}
For $n=4$ and $\frac16\leq T<1$, the competitive ratios are as follows:
\begin{enumerate}
    \item for $\frac16\leq T<\frac15$, $\gamma^{\mathcal I_{4,T}}\approx 0.2578$,
    \item for $\frac15\leq T<\frac14$, $\gamma^{\mathcal I_{4,T}}=2-\sqrt  3\approx 0.2679$,
    \item for $\frac14\leq T<\frac13$, $\gamma^{\mathcal I_{4,T}}\approx 0.2803$,
    \item for $\frac13\leq T<\frac12$, $\gamma^{\mathcal I_{4,T}}=1-\frac{\sqrt 2}2\approx 0.2929$,
    \item for $\frac12\leq T<1$, $\gamma^{\mathcal I_{4,T}}\approx 0.3177$.
\end{enumerate}
\end{claims}

\textit{General Idea of the proof:} If all the requests in figure \ref{decision tree} are maintained, from the delay condition and from the fact that the adversary prevents the vehicle from serving two requests, we obtain the following conditions on the release times: 
\begin{itemize}
    \item $t_4< \t_1=2$,
    \item $t_3\leq t_4-T<2-T$,
    \item $t_2\leq t_3-T<2-2T$,
    \item $t'''_4<\t_2=2t_2<4-4T$,
    \item $t'_3<t'''_4-T<4-5T$,
    \item $t''_4<\t_3\leq t'_3+2-T<6-6T$.
\end{itemize}

In order for the optimal offline algorithm to serve all the requests, we require that $t''_4\geq 5$. Hence, this is possible only when $T<\frac 16$.

If the adversary chooses not to release $f''_4$, the conditions become:
\begin{itemize}
    \item $t_4< 2$,
    \item $t_3\leq t_4-T<2-T$,
    \item $t_2\leq t_3-T<2-2T$,
    \item $t'''_4<\t_2=2t_2<4-4T$,
    \item $t'_3<t'''_4-T<4-5T$,
\end{itemize}

The optimal offline condition now gives $T<\frac 15$.

If the adversary chooses not to release $f_4$, the conditions become:
\begin{itemize}
    \item $t_3<\t_1=2$,
    \item $t_2\leq t_3-T<2-T$,
    \item $t'''_4<\t_2=2t_2<4-2T$,
    \item $t'_3<t'''_4-T<4-3T$,
    \item $t''_4<\t_3\leq t'_3+2-T<6-4T$.
\end{itemize}

The optimal offline condition now gives $T<\frac 14$.

If the adversary chooses not to release $f_4$ and $f''_4$, the conditions become:
\begin{itemize}
    \item $t_3<\t_1=2$,
    \item $t_2\leq t_3-T<2-T$,
    \item $t'''_4<\t_2=2t_2<4-2T$,
    \item $t'_3<t'''_4-T<4-3T$,
\end{itemize}

The optimal offline condition now gives $T<\frac 13$.

If the adversary chooses not to release $f_4$ and $f'''_4$, the conditions become:
\begin{itemize}
    \item $t_3<\t_1=2$,
    \item $t_2\leq t_3-T<2-T$,
    \item $t'_3<\t_2=2t_2<4-2T$,
    \item $t''_4<\t_3\leq t'_3+2<6-2T$,
\end{itemize}

The optimal offline condition now gives $T<\frac 12$.

Solving the corresponding minmax problem in each of these cases yields the above competitive ratios. Note that choosing not to release $f'_4$ does nothing to loosen the conditions on the times. Also, choosing not to release $f'''_4$ gives the condition $T<\frac14$ and yields the same competitive ratio as not releasing $f_4$. Finally, not releasing $f''_4$ and $f'''_4$ gives $T<\frac14$, same as not releasing $f'''_4$, so it is not relevant. \qed

In Table \ref{fig n=4}, we summarise the results obtained for $n=4$.

\begin{table}[ht]
    \centering
        \setcellgapes{3pt}
    \makegapedcells

    \begin{tabular}{|c|c|c|}
        
         \hline
         Delay &  Performance & Competitive Ratio\\
         \hline
         $T<1/6$ & \multirow{4}{*}{$1/4$} & $1/4$\\
         \cline{1-1} \cline{3-3} 
         $1/6\leq T<1/5 $ &  & $0.2578$\\
         \cline{1-1} \cline{3-3} 
         $1/5\leq T<1/4$ &  & $2-\sqrt{3}\approx 0.2679$\\
         \cline{1-1} \cline{3-3} 
         $1/4\leq T<1/3$ &  & $0.2803$\\
         \hline
         $1/3\leq T<1/2$ & $2-\sqrt{3}\approx 0.2679$ & $1-\sqrt 2/2\approx 0.2929$\\
         \hline
         $1/2\leq T< 1$ & $1-\sqrt{2}/2\approx 0.2929$ & $0.3177$\\
         \hline
         $1\leq T< 1.5$ & $1/\varphi^2\approx 0.3820$ & $1/\varphi^2\approx 0.3820$\\
         \hline
         $1.5\leq T< 5/3$ & $1/2$ & $1/2$\\
         \hline
         $5/3\leq T$ & 1 & 1\\
         \hline
    \end{tabular}
    \caption{Performances and Competitive Ratios for $n=4$}
    \label{fig n=4}
\end{table}

\noindent We conjecture that the behaviour of the performance and competitive ratio is the same for $n>4$ as that exhibited for $n=4$:

\begin{conj}
For a fixed $n\geq 4$, the functions $\lambda^{\mathcal I_{n,T}}$ and $\gamma^{\mathcal I_{n,T}}$ are step functions in $T$, constant on the interval $[\frac1k,\frac1{k-1}[$, for $2\leq k\leq 2^{n-1}-2$.
\end{conj}


\section{Generalisation to Geodesic Metric Spaces}\label{sec:geodesic}

In this section, we investigate which of the previous results remain valid when the game is played on a geodesic space instead of a segment.



Recall that the \emph{diameter} of $E=(X,d)$ is defined as $\sup_{(x,y)\in X^2} d(x,y)$.

\begin{prop}\label{negative}
Given a geodesic metric space $E$ of diameter $2$ with two points at distance 2 and given a delay $T$, the performance and competitive ratio on $E$ are at most equal to that on $[-1,1]$. In other words, any negative result concerning $\OPTiWinD$ on a segment can be extended to geodesic metric spaces for which the diameter is reached.
\end{prop}

\begin{proof}
Let $E=(X,d)$ be a geodesic metric space with diameter 2 and $A,B\in X$ such that $d(A,B)=2$. Any strategy of the adversary applicable to $[-1,1]$ can be used on a minimum distance path $\gamma$ from $A$ to $B$. If all requests are located on $\gamma$, since $\gamma$ is a minimum distance path, the vehicle cannot improve its performance by moving outside of $\gamma$. Thus, for all delay $T$, the performance and competitive ratio on $E$ are at most equal to those on $[-1,1]$.
\end{proof}


\begin{theorem}\label{extension}
Theorem \ref{GR_0 optimal}, Theorem \ref{frac 1n} and Theorem \ref{GREE0 superoptimal} can be extended to all geodesic metric spaces of diameter 2.
\end{theorem}

\begin{proof}
Let $E=(X,d)$ be a geodesic metric space with diameter 2. The positive part of these results is given in Lemma \ref{GR_0 perf}, which states that the algorithm $GR_0$ guarantees a performance of $\frac1n$. This is still true when the game is played on $E$.
If the diameter of $E$ is reached, then the result follows from Proposition \ref{negative}. Otherwise, for $\e>0$, we can find $A,B\in X$ such that $d(A,B)=2-\e$ and a path $\gamma$ from $A$ to $B$ of length $2-\e$. Using the same notations as in the proof of Theorem \ref{frac 1n}, we obtain the following conditions:

\begin{itemize}
    \item $t_n<2-\e$,
    \item $t_{n-1}<2-\e-T$,
    \item $t_i<2-\e-2^{n-1-i}T$,
    \item $1+T\leq t_2<2-\e-2^{n-3}T$.    
\end{itemize}
Hence, we obtain $T< (1-\e)T_0$. Thus, choosing $\e<1-\frac T{T_0}$ is sufficient for the strategy described in the proof of Theorem \ref{frac 1n} to work on $\gamma$. Theorem \ref{GR_0 optimal} also works on $\gamma$ as it is a special case of Theorem \ref{frac 1n}.

In the proof of Theorem \ref{GREE0 superoptimal}, all that changes when playing on $\gamma$ is the values of the ETAs, which only affect the proof of point c). However, the relation $\t_{i+1}=\t_i+2-2^{n-1-i}T+O(\epsilon)$ is still true. Thus, for $\e$ small enough, the optimal offline algorithm will serve all requests, and Theorem \ref{GREE0 superoptimal} is still valid in this case.
\end{proof}

\begin{definition}
We say that a geodesic metric space $E=(X,d)$ with origin $O$ and diameter $2$ is \emph{centred} if $X\subset \overline B(O,1)$, where $\overline B(O,1)$ is the closed ball of radius $1$ centred in $O$.
\end{definition}

\begin{theorem}\label{centred extension}
 Theorem \ref{large delay} and Theorem \ref{weighted performance} can be extended to centred geodesic metric spaces of diameter 2.
\end{theorem}

\begin{proof}
Let $E=(X,d)$ be a geodesic metric space with diameter 2 and origin $O$. In the proof of Theorem \ref{large delay}, the positive part of the result is shown by considering algorithm $GR_1$. When playing on $E$, we adapt $GR_1$ by heading towards $O$ instead of $0$ in case (a). The proof that $GR_1$ work on $E$ is identical to that of Theorem \ref{large delay}. The key point is that we still have $i_0=\lfloor \frac{1}{2-T}\big\rfloor$.

In order to prove the negative result, we choose $A,B\in X$ such that $d(A,B)=2-\e'$ and a minimum distance path $\gamma$ from $A$ to $B$. The adversary can apply on $\gamma$ the strategy described in the proof of Theorem \ref{large delay} with one modification: for $i\geq i_0+2$, choose $\eta_i=\frac{i-i_0}n-\frac3{2n}-(i-\frac32)\e'$.  This allows the relation  $t_i=\t_{i-1}-\frac1{2n}$ for $i\geq i_0+2$ to be valid on a path of length $2-\e'$. It follows that for $\e'<\frac5{n^2}$, we still have $\eta_i<1$, $\forall i\geq i_0+2$. Therefore, this strategy of the adversary limits the competitive ratio of the vehicle to $\alpha_{n-\lfloor \frac{1}{2-T}\rfloor}$.

For Theorem \ref{weighted performance}, we choose $A$ and $B$ with $d(A,B)> 2T$. In the proof of Theorem \ref{weighted performance}, $-1$ and $1$ are replaced with $A$ and $B$, respectively, and the interval $]T-1,1-T[$ becomes $X-(\overline B(A,T)\cup \overline B(B,T))$. The rest of the proof remains identical. 
\end{proof}

\begin{remark}
If $E$ is a circle, $E$ is not centred, and both Theorem \ref{large delay} and Theorem \ref{weighted performance} do not apply whenever $n\geq 3$. 
\end{remark}

\begin{prop}
Depending on the geodesic metric space, the smallest delays for which the optimal performance and optimal competitive ratio are greater than $\frac1n$ may be greater than $T_0=\frac1{2^{n-3}+1}$ and $T_1=\frac1{2^{n-1}-2}$ respectively.
\end{prop}

\begin{proof}
Let $E$ be a star with central point $O$ and $3$ branches of length $1$. Lemma \ref{2 fires} and Lemma \ref{Al2} do not apply to $E$ as the adversary can now release requests at the end of the third branch. It follows that Theorem \ref{Al1} and Theorem \ref{T1} do no apply when the game is played on $E$.
Let $A,B$ and $C$ denote the endpoints of the branches of the star. For $n=5$ and $T<\frac14$, let $\e<1-4T$. The adversary will release a request $f_1=(A,[1,3],1)$. Recall that when the vehicle may choose between two symmetric requests of equal weights, it is optimal to head towards the closest one. So, while the vehicle keeps going towards $A$, the adversary will release the following requests: $f_2=(B,[2-3T-\e, 4-3T-\e],1)$, $f_3=(C,[2-2T-\e,4-2T-\e],1)$, $f_4=(C,[2-T-\e,4-T-\e],1)$ and $f_5=(C,[2-\e,4-\e],1)$. With this strategy, the adversary limits the optimal performance to $\frac15$, even if $T_0<T<\frac14$. Note that this strategy was not possible on the segment because it requires that $f_2$ and $f_3$ be located in different branches of the star, otherwise the vehicle could reach both. A similar strategy can be used by the adversary to limit the optimal competitive ratio to $\frac15$ when $T$ is slightly greater than $T_1$.
\end{proof}

\section{Conclusion}
 
 In this paper, we introduced the delay between requests as a new parameter in the standard online orienteering problem. While this new parameter seems very natural, to our knowledge, it had not been studied previously. We analysed the performances and competitive ratios in the case where the length of the time windows is equal to the diameter of the space. We obtained a complete resolution when the number of requests is at most 4. In the case of $n$ requests, we solved the problem when $T\geq 1$ or $T< \frac1{2^{n-1}-2}$. Our results for small numbers of requests give us an accurate idea of what to expect in the intermediate case. Other choices regarding the length of the time windows, which may be relevant for different applications, remain to be investigated.
 
\section*{Acknowledgements}
\noindent We acknowledge the support of GEO-SAFE,  H2020-MSCA-RISE-2015 project \# 691161.

\section*{References}
\bibliographystyle{apalike}
\bibliography{Finalversion_TCS}

\newpage
\section*{Appendix}

\vspace{.3cm}

In this section, we study the sequence $(\alpha_n)$ defined in Section \ref{T geq 1}. We give an explicit formula for $\delta_i$ and calculate the limit of $(\alpha_n)$.

We have the following definition of $\alpha_n$  for $n\geq 1$:
\begin{equation*}\label{alpha n}
    \alpha_n=\inf\limits_{\delta \in (\R_+^*)^{n}} \max \{\frac{\delta_1}{S_2}, \cdots , \frac{\delta_{i}}{S_{i+1}},\cdots, \frac{\delta_{n-1}}{S_{n}},\frac{\delta_{n}}{S_{n}}\},
\end{equation*}
where $S_i=\sum_{k=1}^i \d_k$.
Since multiplying all the $\d_i$ by a positive constant makes no difference to the ratios, we may choose $\d_1=1$. To lighten the notation, we will omit the index $n$ and will denote $\alpha_n$ by $\alpha$. 

\begin{prop}\label{delta i}
The value of $\alpha$ is realised for a unique vector $\d\in(\R_+^*)^{n}$ with $\d_1=1$ such that for $1\leq i\leq n$:
$$\delta_i=\Big(\frac 12+\frac{1-2\alpha}{2\sqrt{1-4\alpha}}\Big)\Big(\frac {1+ \sqrt{1-4\alpha}}{2\alpha}\Big)^{i-1}~~~~~~~~~~~~~~~~$$
$$~~~~~~~~~~~~+\Big(\frac 12-\frac{1-2\alpha}{2\sqrt{1-4\alpha}}\Big)\Big(\frac {1- \sqrt{1-4\alpha}}{2\alpha}\Big)^{i-1}.$$
\end{prop}

\begin{proof}

Note that:
\begin{eqnarray*}
\lefteqn{\inf\limits_{\delta_{n-1},\delta_{n}} \max \{\frac{\delta_{n-2}}{S_{n-1}}, \frac{\delta_{n-1}}{S_{n}},\frac{\delta_{n}}{S_{n}}\}}~~~~~~~~~~~~~~~~~~~~\\
&=& \inf\limits_{\delta_{n-1}}\inf\limits_{\delta_{n}} \max \{\frac{\delta_{n-2}}{S_{n-1}},\max \{ \frac{\delta_{n-1}}{S_{n}},\frac{\delta_{n}}{S_{n}}\}\}\\ &=& \inf\limits_{\delta_{n-1}} \max \{\frac{\delta_{n-2}}{S_{n-1}},\inf\limits_{\delta_{n}} \max \{ \frac{\delta_{n-1}}{S_{n}},\frac{\delta_{n}}{S_{n}}\}\}.
\end{eqnarray*}

By repeating this inversion, we obtain that:
\begin{eqnarray*}
\lefteqn{\alpha= \inf\limits_{\delta_2} \max \{\frac{1}{S_2},\inf\limits_{\delta_3} \max \{\frac{\delta_2}{S_3},\inf\limits_{\delta_4}\max \{ }~~~~~~~~~~~~~~~~~~~~~~~~~~~~~\\
& &\cdots, \inf\limits_{\delta_{n}} \max \{\frac{\delta_{n-1}}{S_{n}},\frac{\delta_{n}}{S_{n}}\}\cdots\}\}\}  
\end{eqnarray*}

In the above equation, the operator $\inf_{\delta_i}$ is applied to the $\max$ of two terms, the first being decreasing in $\delta_i$ and the second increasing in $\delta_i$. It follows that the optimal value is reached when:
\begin{equation*}\label{alpha}
\forall 1\leq i \leq n-1, ~ \frac{\delta_i}{S_{i+1}}= \frac{\delta_{n}}{S_{n}}=\alpha
\end{equation*}
Thus, we obtain the following equations $(E_i)$, for all $2 \leq i \leq n$:
$$(E_i)~~ \alpha=\frac{\delta_i}{S_{i+1}}=\frac{\delta_{i-1}}{S_{i}} ~~\mathrm{and~~} (E_n)~~\d_{n-1}=\d_n=\alpha S_n.$$

\noindent Thus, $\delta_i=\alpha S_{i+1}$ and $\delta_{i-1}=\alpha S_i$. Hence, $\delta_i -\delta_{i-1}=\alpha\delta_{i+1}$. The sequence $(\delta_i)$ satisfies a second order linear recurrence relation. Its characteristic equation is $X^2-\frac1{\alpha}X+\frac 1\alpha=0$ and the discriminant is $\Delta=\frac1{\alpha^2}(1-4\alpha)$.

If $\Delta=0$, we have $\alpha=\frac14$. Hence, $\d_i=(i+1)2^{i-2}$, for $1\leq i\leq n$. This contradicts $(E_n)$. So $\Delta\neq 0$.

Thus the characteristic equation has two roots $r_{\pm}=\frac1{2\alpha} (1\pm \sqrt{1-4\alpha})$, where $\sqrt{1-4\alpha}$ may be an imaginary number. So there exists a $\lambda$ and a $\mu$ such that $\forall 1 \leq i \leq n,  \delta_i=\lambda r_+^{i-1}+\mu r_-^{i-1}$.


 Since $\d_1=1$, we have $\lambda+\mu=1$. Similarly, using $\d_2$, we obtain that $\lambda-\mu= \frac{2}{\sqrt{\Delta}}(\d_2-\frac1{2\alpha})$. The result follows.
\end{proof}

\begin{lemma}\label{Delta<0}
The roots of the characteristic equation $X^2-\frac1\alpha X +\frac1\alpha=0$ are complex conjugates.
\end{lemma}

\begin{proof}

Using the same notations as above, let us assume that $\Delta> 0$ (i.e. $\alpha < \frac 14$). Then, $r_+$, $r_-$, $\lambda$ and $\mu$ are real numbers. It follows from $(E_1)$ that $\d_2=\frac1\alpha -1$. Hence, $\lambda -\mu>0$. Since $\lambda+\mu=1$, we have $\lambda>0$.
Since  $\sqrt{1-4\alpha}<1$, we have $r_+\geq r_->0$.

Since $\delta_{n-1}=\delta_{n}$, we have:
    \begin{equation}\label{lmr+-}
      \lambda r_+^{n-2}(r_+-1)=-\mu r_-^{n-2}(r_--1)  
    \end{equation}
Since $\alpha < \frac 14$, we have $r_+\geq 2$. It follows that in Equation \ref{lmr+-}, the left hand side is positive and $\mu\neq 0$. Let us now consider two cases:
\begin{enumerate}
    \item  If $\mu<0$, it follows from Equation \ref{lmr+-} that $r_--1>0$, since all the other terms are positive. Yet, $\lambda>-\mu$ and $r_+\geq r_-$. Hence, Equation \ref{lmr+-} is impossible.

    \item If $\mu>0$, the right hand side of Equation \ref{lmr+-} being positive requires $r_-<1$. Since $r_+\geq 2$ and $r_->0$, we have $r_+-1>1-r_-$. Yet, since $\lambda \geq \mu$, Equation \ref{lmr+-} is impossible.
\end{enumerate}

Therefore, since $\Delta\neq 0$, we have $\Delta<0$ and $\alpha>\frac 14$. Hence, $r_+$ and $r_-$ are complex conjugates, as are $\lambda$ and $\mu$.

\end{proof}

\begin{lemma}
The sequence $(\alpha_n)$ is decreasing.
\end{lemma}

\begin{proof}
Choosing $(\d_1,\ldots,\d_n)$ as the vector which realises $\alpha_{n-1}$ and $\d_n=0$ yields: \\$\max \{\frac{\delta_1}{S_2}, \cdots , \frac{\delta_{i}}{S_{i+1}},\cdots, \frac{\delta_{n-1}}{S_{n}},\frac{\delta_{n}}{S_{n}}\}=\alpha_{n-1}$. Hence, $\alpha_n\leq \alpha_{n-1}$. Since this is not the vector described in Proposition~\ref{delta i}, we have $\alpha_n<\alpha_{n-1}$.
\end{proof}

\begin{prop}\label{alpha decreases}
The sequence $(\alpha_n)$ decreases towards $\frac14$ when $n\rightarrow +\infty$. 
\end{prop}

\begin{proof}

Using the same notations as above, it follows from Proposition \ref{delta i} and Lemma \ref{Delta<0} that $\delta_i=2\Re(\lambda r_+^i)>0$. Hence, $-\frac\pi 2<\arg(\lambda r_+^i)<\frac\pi 2$. Thus, $$i.\arg r_+ +\arg \lambda \mod (2\pi)<\frac\pi 2.$$

Since $\lambda +\mu=1$, we have $\Re(\lambda)=\frac 12$, and $-\frac\pi 2<\arg \lambda< \frac \pi 2$. Hence, $0<i.\arg r_+ \mod (2\pi)<\pi $. Thus, there exists an integer $k_i$, $0\leq k_i<i$, such that $\frac{2k_i\pi}i< \arg r_+<\frac{(2k_i+1)\pi}i$. Since this is true for all $1\leq i\leq n$, we have $k_i=0$ for all $i$, and $0<\arg r_+<\frac\pi{n}$. Therefore $\lim_{n\rightarrow +\infty} \arg r_+=0$. Since $\Im(r_+)=\frac{\sqrt{4\alpha-1}}{2\alpha}$, we have $\lim_{n\rightarrow +\infty} \alpha=\frac 14$.

\end{proof}

\begin{prop}
The first values of $(\alpha_n)$ are $\alpha_1=1$, $\alpha_2=\frac12$, $\alpha_3=\frac1{\varphi^2}$, where $\varphi$ is the golden ratio, and $\alpha_4=\frac13$.
\end{prop}

\begin{proof}
While $\alpha_1$ is trivially equal to $1$, the following terms of the sequence are calculated using the equations $(E_i)$.
\begin{itemize}
    \item For $n=2$, $(E_2)$ yields $\d_1=\d_2=1$ and $\alpha=\frac12$.
    \item For $n=3$, equations $(E_2)$ and $(E_3)$ yield $\frac1{S_2}=\frac{\d_2}{S_3}=\frac{\d_3}{S_3}=\alpha$. Thus, $\d_2=\d_3=\varphi$ and $\alpha=\frac1{\varphi^2}$.
    \item For $n=4$, equations $(E_2)$, $(E_3)$ and $(E_4)$ yield $\frac1{S_2}=\frac{\d_2}{S_3}=\frac{\d_3}{S_4}=\frac{\d_4}{S_4}=\alpha$. Thus $\d_2=2$ and $\d_3=\d_4=3$ and $\alpha =\frac13$.
\end{itemize}
\end{proof}

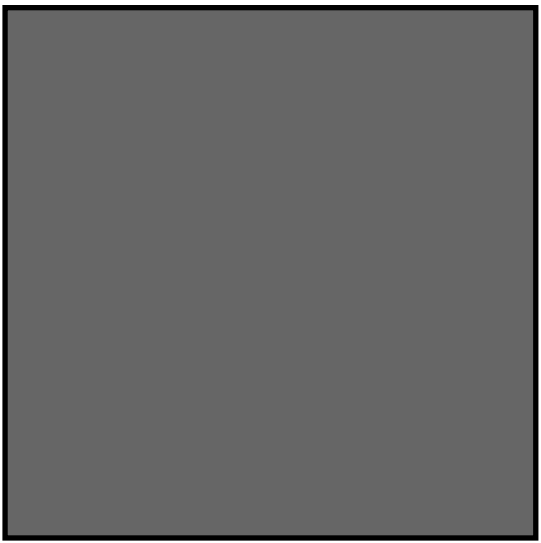
\end{document}